\documentclass[aps,prl,showpacs,twocolumn,floatfix]{revtex4-1}
\usepackage{amsmath}
\usepackage{amsfonts}
\usepackage{amssymb}
\usepackage{color}
\usepackage{datetime}
\usepackage{flushend}
\usepackage{graphicx}
\usepackage[colorlinks=true,allcolors=black]{hyperref}
\usepackage{times}

\usepackage[makeroom]{cancel}
\usepackage{layouts}

\newenvironment{proof}%
{\vspace{\baselineskip}\par\noindent{\textbf{Proof}\ }}%
{\QED}

\renewcommand{\@}{\partial}
\providecommand{\abs}[1]{\left\lvert #1 \right\rvert}

\newcommand{\cconj}[1]{{#1}^*}

\renewcommand{\d}{\mathrm{d}}

\newcommand{\dff}[3]{\frac{\partial^2 #1}{\partial #2\partial #3}}
\newcommand{\E}[1]{\mathrm{E}\left[#1\right]}			
\newcommand{\eq}[1]{(\ref{#1})}
\def\eqtwo(#1,#2){(\ref{#1},\ref{#2})}

\newcommand{\fig}[1]{fig.~\ref{#1}}
\newcommand{\Fig}[1]{Fig.~\ref{#1}}

\newcommand{\intinf}{\int\limits_{-\infty}^{\infty}}

\newcommand{\mean}[1]{\overline{#1}}
\newcommand{\mes}[2]{\mathop{\mathrm{Vol}}_{#1}#2}
\providecommand{\norm}[1]{\left\lVert#1\right\rVert}
\renewcommand{\O}[1]{O\left(#1\right)}				
\renewcommand{\P}[1]{\mathbb{P}\left[#1\right]}			
\newcommand{\qed}{\rule{1ex}{1ex}\vspace{0mm}}
\newcommand{\QED}{\hfill\qed\par}

\newcommand{\Real}{\mathbb{R}}
\newcommand{\Torus}{\mathbb{T}}

\newcommand{\summ}{\sideset{}{'}\sum}				

\newcommand{\Span}{\mathop{\mathrm{span}}}                      
\newcommand{\Zahlen}{\mathbb{Z}}

\ifx\notcolor\undefined\newcommand{\notcolor}{blue}\else\fi
\newcommand{\+}[2]{\def#1{{\color{\notcolor}#2}}}
\newcommand{\1}[2]{\def#1##1{{\color{\notcolor}#2}}}

\+{\A}{A}		
\+{\a}{a}		
\+{\ai}{\alpha'} \+{\aii}{\alpha''} 
\+{\ap}{\alpha_+} \+{\am}{\alpha_-} 
\+{\al}{\alpha}         
\+{\arbF}{\phi} \+{\arbG}{\psi} 
\+{\Ball}{B}		
\+{\mesBall}{\beta}	
\+{\C}{C}		
\+{\D}{D}		
\+{\devn}{\Delta}	
\+{\dfreqq}{\delta}	
\+{\eps}{\epsilon}      
\+{\F}{F}		
\+{\f}{f}		
\+{\fx}{\eta}		
\+{\Freqset}{N}         
\+{\freq}{\omega}	
\+{\freqn}{\freq_{\m+1}}        
\+{\freqx}{\tilde{\freq}}       
\+{\G}{\mathcal{G}}     
\+{\g}{g}		       
\+{\gx}{\xi}	       
\+{\I}{I}               
\+{\j}{j}		        
\+{\K}{k}		
\+{\l}{\ell}               
\+{\m}{m}		
\+{\mq}{k}		
\1{\meansizeT}{\mu_{#1}}        
\+{\mspath}{\varpi_{\infty}}    
\1{\mspathT}{\varpi_{#1}}       
\+{\mspathsq}{\varpi_{\infty}^2}        
\1{\mspathsqT}{\varpi_{#1}^2}   
\+{\mssize}{\sigma_{\infty}}    
\1{\mssizeT}{\sigma_{#1}}       
\+{\mssizesq}{\sigma_{\infty}^2}        
\1{\mssizesqT}{\sigma_{#1}^2}   
\+{\N}{\mathbf{n}}      
\+{\n}{n}	        
\+{\ni}{\n'}	        
\+{\nii}{\n''}	        
\+{\param}{\lambda}     
\+{\phase}{\theta}	
\+{\phasen}{\phase_{\m+1}}  
\+{\phasex}{\tilde{\phase}} 
\+{\sig}{\chi}          
\+{\T}{T}	        
\+{\t}{t}	        
\+{\tipp}{p}		
\+{\mtipp}{\mean{\tipp}}	
\+{\tipangle}{\varphi}  
\+{\Tipangle}{\Phi}	
\+{\Tiprot}{h}   	
\+{\vvv}{\tilde v}   	
\+{\tiprot}{\Omega}     
\+{\tipx}{\tipp_x}	
\+{\tipy}{\tipp_y}	
\+{\Tipv}{v}		
\1{\Tipvr}{v^r_{#1}}    
\1{\Tipvi}{v^i_{#1}}    
\+{\tipv}{V}		
\+{\tipvx}{\tipv_x}	
\+{\tipvy}{\tipv_y}	
\+{\termD}{D}		
\+{\termM}{M}		
\+{\termP}{P}		
\+{\termS}{S}		
\+{\r}{\mathbf{r}}	
\+{\veloc}{s}		
\+{\Uspace}{\mathcal{U}}
\+{\u}{u}		
\+{\um}{\tilde{u}}	
\+{\w}{W}		
\+{\X}{X}               
\+{\Xspace}{\mathcal{X}}
\+{\x}{x}               
\+{\z}{z}               
\+{\zi}{\z'}            
\+{\zii}{\z''}          
\+{\ziii}{\zeta}        

\newcommand{\myfigure}[3]{\begin{figure}[htb] #1 \caption[]{#2} \label{#3} \end{figure}}
\newcommand{\dblfigure}[3]{\begin{figure*}[htb] #1 \caption[]{#2} \label{#3} \end{figure*}}

\newtheorem{proposition}{Proposition}

\newcommand{\Quoziente}{\mathbb{Q}}

\begin{document}

\title{St. Petersburg paradox for quasiperiodically hypermeandering spiral waves}
\author{V. N. Biktashev}
\affiliation{Department of Mathematics, 
University of Exeter, Exeter EX4 4QF, UK}
\author{I. Melbourne}
\affiliation{Mathematics Institute, 
University of Warwick, Coventry CV4 7AL, UK}
\date{\yyyymmdddate\today\ \currenttime}

\makeatletter
\typeout{author=\meaning\@author}
\makeatother

\begin{abstract}
It is known that quasiperiodic hypermeander of spiral waves almost
certainly produces a bounded trajectory for the spiral tip.  We analyse
the size of this trajectory.  We show that this deterministic question
does not have a physically sensible deterministic answer and requires
probabilistic treatment.  In probabilistic terms, the size of the
hypermeander trajectory proves to have an infinite expectation, despite
being finite with probability one.
This can be viewed as a physical manifestation of the classical
  ``St. Petersburg paradox'' from probability theory and economics.
\end{abstract}


\maketitle

Rotating spiral waves are a class of self-organized patterns observed
in a large variety of spatially extended thermodynamically nonequilibrium
systems with oscillatory or excitable local dynamics, of
physical, chemical or biological nature~\cite{%
  Zhabotinsky-Zaikin-1971,%
  Allessie-etal-1973,%
  Alcantara-Monk-1974,%
  Carey-etal-1978,%
  Gorelova-Bures-1983,%
  Murray-etal-1986,%
  Schulman-Seiden-1986,%
  Madore-Freedman-1987,%
  Jakubith-etal-1990, %
  Lechleiter-etal-1991,%
  Frisch-etal-1994,%
  Yu-etal-1999,%
  Agladze-Steinbock-2000,%
  Kastberger-etal-2008%
}. Of particular practical importance are spiral waves of electrical
excitation in the heart muscle, where they underlie dangerous
arrhythmias~\cite{%
  Alonso-etal-2016-PR%
}. Very soon after their experimental discovery in
Belousov-Zhabotinsky reaction, it was noticed that rotation of spiral
waves is not necessarily steady, but their tip can describe a
complicated trajectory, ``meander''~\cite{%
  Winfree-1973%
}. Subsequent mathematical modelling allowed a more detailed
classification of possible types of rotation of spiral waves in ideal
conditions: steady rotation like a rigid body, when the tip of the
spiral travels along a perfect circle; meander, when the solution is
two-periodic and the tip traces a trajectory resembling a roulette
(hypocycloid or epicyloid) trajectory; and more complicated patterns, dubbed
``hypermeander''\cite{%
  Roessler-Kahlert-1979,%
  Zykov-1986,%
  Winfree-1991%
}. Often different types of meander may be observed in the same model
at different values of parameters~\cite{Winfree-1991}, including
cardiac excitation models
(see~\fig{patterns}).
The question of the
spatial extent of the spiral tip path can be of practical
importance. Here we discuss this question for quasiperiodic hypermeander.

\myfigure{\includegraphics{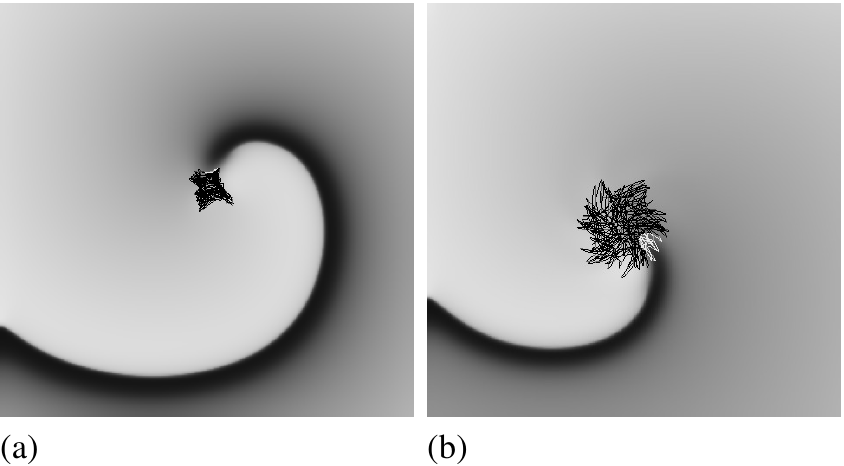}}{%
  Snapshots of anticlockwise rotating spiral waves of electrical
  excitation, together with traces of their tips, in a
  reaction-diffusion model of guinea pig ventricular tissue, (a)
  classical meander in a model with standard
  parameters~\cite{Biktashev-Holden-1996}, (b) hypermeander in the
  same model with parameters changed to represent Long QT
  syndrome~\cite{Biktashev-Holden-1998a}.
}{patterns}

\paragraph{The equations of motion of the meandering spiral tip}
may be derived by the standard procedure of rewriting the underlying partial
differential equations as a skew product~\cite{%
  Barkley-1994,%
  Fiedler-etal-1996,
  Biktashev-etal-1996,
  Sandstede-etal-1997,
  Biktashev-Holden-1998,
  Golubitsky-etal-2000,
  Nicol-etal-2001,%
  Ashwin-etal-2001, 
  Roberts-etal-2002,
  Beyn-Thummler-2004,
  Foulkes-Biktashev-2010,
  Hermann-Gottwald-2010,
  Gottwald-Melbourne-2013
}.  Consider the $\l$-component
reaction-diffusion system on the plane, 
\[
  \partial_\t \u = \D\nabla^2\u+\f(\u), 
  \quad \u(\r,\t)\in\Real^\l,
  \quad \r\in\Real^2,
\]
as a flow in the phase space which
is an infinite-dimensional space of functions
$\Real^2\to\Real^\l$.  The symmetry group is the Euclidean group $\G$ of
transformations of the plane $\g:\Real^2\to\Real^2$ acting on $\Real^2$ by translations
and rotations and thereby acting on functions $\u:\Real^2\to\Real^\l$ by
$\u(\r)\mapsto \u(\g^{-1}\r)$.

Such systems with symmetry, or ``equivariant dynamical systems'' can be cast
into a \emph{skew product} form
\[
  \dot \X = \fx(\X), \qquad \dot \g = \g\gx(\X),
\]
on $\Xspace\times \G$, where the dynamics on the symmetry group $\G$
is driven by the ``shape dynamics'' on a cross-section $\Xspace$
transverse to the group directions. Here, $\g\gx(\X)$ denotes the
action of the group element $\g\in \G$ on vectors $\gx(\X)$ lying in
the Lie algebra of $\G$; $\fx$ and $\gx$ are defined by
components of the vector field along $\Xspace$ and orbits of $\G$
respectively.

The shape dynamics $\dot \X=\fx(\X)$ on the cross-section $\Xspace$
is   a  dynamical system
devoid of symmetries.  Substituting the solution $\X(\t)$ for the
shape dynamics into the $\dot \g$ equation yields the nonautonomous
finite-dimensional equation $\dot \g=\g\gx(\X(\t))$ to be solved for
the group dynamics.

For the Euclidean group $\G$ consisting of planar
translations $\tipp$ and rotations $\tipangle$, the equations become
\begin{equation}				\label{tipt}
  \dot \X = \fx(\X), \qquad
  \dot\tipangle = \Tiprot(\X) , \qquad   		
  \dot\tipp = \Tipv(\X) \; e^{i\tipangle} .	
\end{equation}
The variables $\tipp$ and $\tipangle$ can be interpreted as position and
orientation of the tip of the spiral, then $\X(\t)$ describes
the evolution in the frame comoving with the
tip~\cite{Biktashev-etal-1996,Foulkes-Biktashev-2010}.
Standard low-dimensional attractors in $\Xspace$ produce
the classical tip meandering patterns through the $\dot\tipp$ equation, namely an
equilibrium produces stationary rotation, a limit cycle produces the
two-frequency flower-pattern meander, and more complicated attractors
produce ``hypermeander''.  Hypermeander produced by chaotic base
dynamics is asymptotically a deterministic Brownian motion
\cite{Biktashev-Holden-1998,Nicol-etal-2001}.  Quasiperiodic base
dynamics produce another kind of
hypermeander, with tip trajectories almost certainly bounded, but
exhibiting unlimited directed motion at a dense set of parameter
values \cite{Nicol-etal-2001}. Similar dynamics may be observed
  when a spiral with two-periodic meander is subject to periodic
  external forcing~\cite{Mantel-Barkley-1996}.

Our aim is to characterize the size of a quasiperiodic meandering trajectory
when it is finite.  

\paragraph{The mathematical problem.} 

We assume $\m$-frequency quasiperiodic dynamics in the base system, $\m\ge2$,
with $\X=\phase\in\Torus^\m=(\Real/2\pi\Zahlen)^\m$ being coordinates on the invariant
$\m$-torus, so that the shape dynamics $\dot\X=\fx(\X)$ becomes
\begin{align}
\dot \phase = \freq ,				\label{qp}
\end{align}
where $\freq\in\Real^\m$ is a set of irrationally related
frequencies~\footnote{%
  As discussed earlier, various types of spiral behaviour are
  associated with various types of dynamics (steady-state, periodic,
  quasiperiodic, chaotic) in the base equation $\dot\X=\fx(\X)$.  All
  these types of dynamics are known to occur with positive
  probability.  In particular, KAM theory predicts the existence of
  quasiperiodic dynamics.  In this paper, we take the point of view
  that the base dynamics is known to be quasiperiodic (in accordance
  with the observations
  in~\cite{Roessler-Kahlert-1979,Zykov-1986,Winfree-1991} and analyse
  the consequent behaviour in the full system of equations. %
}. 
The $\dot \g$ equations become
\begin{equation}
      \dot\tipangle = \Tiprot(\phase) , \qquad   		
      \dot\tipp = \Tipv(\phase) \; e^{i\tipangle} .	
\label{extension}
\end{equation}
Equations (\ref{qp},\ref{extension}) comprise a closed 
system describing the trajectory of the quasiperiodic meandering spiral tip.

\paragraph{The $\Real^1$-extension of the quasiperiodic dynamics.} First we
illustrate our main idea for the simpler case where the orientation angle $\tipangle$ is absent and the position $\tipp$ is one-dimensional. 
The shape dynamics remains as in~(\ref{qp}) with $\phase\in\Torus^\mq$, $\mq\ge2$.
Then a point with coordinate 
$\tipp\in\Real^1$ moves according to
\begin{equation}
\dot\tipp = \veloc(\phase)
  = \sum\limits_{\n\in\Zahlen^\mq} \veloc_\n e^{i \n\cdot\phase },	\qquad 
\dot\phase = \freq .					\label{xeqn}
\end{equation}
Termwise
integration gives
\[
\tipp(\t) = \tipp(0)+\veloc_0 \t + \summ\limits_{\n\in\Zahlen^\mq} \frac{-i\veloc_\n}{\n\cdot\freq} \left(e^{i\n\cdot\freq\t}-1\right) ,
\]
where the prime denotes summation over $\n\ne0$.
Consider the infinite sum here, defining the deviation of $\tipp$ from steady motion,
$\devn_\t(\freq)=\tipp(\t)-\tipp(0)-\veloc_0 \t$. 
For an arbitrarily chosen $\freq$, its components are almost certainly
incommensurate, and, moreover, Diophantine. So the denominators in the infinite sum are
nonzero, but many of them are very small; nevertheless they decay slowly with
$\norm{\n}=\left(\n_1^2+\dots+\n_\mq^2\right)^{1/2}$.
 This is compensated by the
fact that if the function $\veloc(\phase)$ is sufficiently smooth, its
Fourier coefficients $\veloc_\n$ in the numerators quickly decay
with $\norm{\n}$.
As a result, the infinite sum remains bounded for
$\t\ge0$, for $\veloc(\phase)$ sufficiently smooth and almost all
$\freq$~\cite{Nicol-etal-2001}.

So if we consider the trajectories in the frame moving with the velocity
$\veloc_0$, we know they are typically confined to a finite space. Now we ask
how large they can be.  The size of a finite piece of trajectory may be
measured in various ways, say by the departure from the initial point
$\devn_\t(\freq)=\tipp(\t)-\tipp(0)$, its time average,
$\meansizeT{\T}(\freq)=\T^{-1}\int_0^\T\devn_\t(\freq)\,\d\t$, and the
corresponding variance,
$\mssizesqT{\T}(\freq)=\T^{-1}\int_0^\T\abs{\devn_\t(\freq)-\meansizeT{\T}(\freq)}^2\,\d\t$.
For instance, as $\T\to\infty$ we obtain
\begin{equation}
\mssizesq(\freq) = \summ\limits_{\n\in\Zahlen^\mq} 
	\frac{\abs{\veloc_{\n}}^2}{(\n\cdot\freq)^2} .	\label{Delta-val}
\end{equation}
By the above arguments, for almost any vector $\freq$, this expression is
finite.  However, as typically all $\veloc_\n$ are nonzero, expression
\eq{Delta-val} is infinite for all $\freq$ for which the denominator is zero,
and for $\mq\ge2$, this is a dense set.  That is,
the function $\mssize(\freq)$ is almost everywhere defined and finite, but is
everywhere discontinuous. The latter property implies that for any physical
purpose, questions about the value of the function at a
particular point are
meaningless, as any uncertainty in the arguments, no matter how small, causes
a non-small, in fact infinite, uncertainty in the value of the function.

Hence, a deterministic view on the function $\mssize(\freq)$ is
inadequate, and we are forced to adopt a probabilistic view. Suppose
we know $\freq$ approximately, say, its probability density is
uniformly distributed in $\Ball=\Ball_{\dfreqq}(\freq_0)$, a ball of
radius $\dfreqq$ centered at $\freq_0$~\footnote{%
  When quasiperiodicity arises via KAM theory, near onset,
  phaselocking leads to a complicated structure for the positive
  measure set $\Freqset$ of frequencies $\freq$ corresponding to
  quasiperiodic dynamics.  The integration should then be over
  $\Freqset\cap\Ball$ rather than the whole ball $\Ball$.  However,
  writing $\Freqset_\param$ to indicate dependence on a parameter
  $\param\to0$, we have in this situation that
  $\mes{\mq}{(\Freqset_\param\cap\Ball)}\to \mes{\mq}{\Ball}$ and hence
  $\lim_{\param\to0}\int_{\Freqset_\param\cap\Ball}
  \abs{\n\cdot\freq}^{-1}\,\d\freq=+\infty$.  So we obtain the same
  conclusion in the limit as $\param\to0$ as before.
}. The expectation of the
trajectory size is then
\[
\E{\mssize}
= \frac{1}{\mesBall} \int\limits_{\Ball} \mssize(\freq) \;\d\freq
= \frac{1}{\mesBall} \int\limits_{\Ball} \left(
  \summ\limits_{\n\in\Zahlen^\mq}  \frac{\abs{\veloc_{\n}}^2}{(\n\cdot\freq)^2} 
  \right)^{1/2} \d\freq ,				
\]
where $\mesBall=\mes{\mq}{(\Ball)}$.
The set of hyperplanes $\n\cdot\freq=0$, $\n\in\Zahlen^\mq$ is dense
so there is an infinite set of $\n\in\Zahlen^\mq$ whose hyperplanes $\n\cdot\freq=0$
cut through $\Ball$. 
For any such $\n$, we have
\[
\E{\mssize}
\geq \frac{1}{\mesBall} \int\limits_{\Ball} 
    \abs{ \frac{\veloc_{\n}}{\n\cdot\freq} } \d\freq.
\]
Then, for some $\A,\eps>0$ depending on $\n$, we have
\[
  \int_{\Ball} \frac{\d\freq}{|\n\cdot\freq|}>\A\int_{-\eps}^{\eps}\,\frac{\d{\z}}{|\z|}=+\infty.
\]
Typically, $\abs{\veloc_{\n}}>0$ for all such $\n$, therefore
we have $\E{\mssize}=+\infty$. 

That is, the deviation from steady motion is almost certainly finite, but its 
average expected value is infinite.

\paragraph{The quasiperiodic hypermeander trajectories.} 

We now return to the equations (\ref{qp},\ref{extension}) governing quasiperiodic hypermeander.
Consider first the $\dot\phase$, $\dot\tipangle$ subsystem
\begin{equation} 
\dot\phase = \freq , \qquad
\dot\tipangle = \Tiprot(\phase) .  \label{rot}
\end{equation}
This has the form of \eq{xeqn} with 
$\mq=\m$,
$\tipp=\tipangle$,
$\veloc=\Tiprot$.
Proceeding as for $\Real^1$-extensions, we obtain
$\tipangle=\tipangle_0+\Tiprot_0 \t + \Tipangle(\phase)$,
where $\Tipangle(\phase) = - i \summ_{\n\in\Zahlen^\m} \Tiprot_\n
\left(e^{i\n\cdot\phase}-1\right)/\n\cdot\freq$.
Substituting into the $\dot\tipp$ equation, we obtain
\[
\dot\tipp=\Tipv(\phase)e^{i\tipangle_0+\Tipangle(\phase)}e^{i\Tiprot_0 \t}
=\Tipv(\phase)e^{i(\tipangle_0+\Tipangle(\phase)+\phase_{\m+1})},
\]
where $\phase_{\m+1}\in\Torus^1$ satisfies the equation
$\dot\phase_{\m+1}=\Tiprot_0$.
Hence the evolution of $\dot\tipp$ is governed by the skew product equations
\begin{equation}
\dot\phasex = \freqx  , \qquad \dot\tipp = \vvv(\phasex) ,	\label{transl}
\end{equation}
where $\freqx=(\freq, \Tiprot_0)\in\Real^{\m+1}$,
$\phasex=(\phase,\phasen)\in\Torus^{\m+1}$ and
\begin{equation}
\vvv(\phasex)=\Tipv(\phase) \, e^{i(\tipangle_0+\Tipangle(\phase)+\phasen)}.
									\label{translv}
\end{equation}
System~\eq{transl} has a similar form to \eq{xeqn}
(separately for the real and imaginary parts of $\tipp$),
except that 
now $\mq=\m+1$,
$\phase\in\Torus^\m$ and $\phasex\in\Torus^{\m+1}$
Also, we notice that due to \eq{translv},
Fourier components $\vvv_\n$ are nonzero only for $\n_{\m+1}=\pm1$,
which implies that $\vvv_0=0$. 
Physically speaking, due to the rotation of the meandering tip, its average spatial
velocity is always zero.
Hence, the function $\mssize(\freqx)$ in this case is just the size of the
trajectory, defined as the root mean square of the distance of the tip
from the centroid of the trajectory.

Based on the results of the previous paragraph, we conclude from here
our main result: for hypermeandering spirals, the long-term average of
the displacement of the tip from its centroid is a random
quantity, which takes finite values with probability one, but has an
infinite expectation. This result is proved rigorously
in~\cite{Melbourne-Biktashev-2018}. The rest of our results below
are at the
physical level of rigour. 

\paragraph{The asymptotic distribution of the trajectory size}
\typeout{distribution}
is fairly generic for typical systems. 
Consider when the trajectory size
\begin{eqnarray}
&&
\mssize(\freqx) 
 = \left(
  \summ\limits_{\n\in\Zahlen^{\m+1}}  \frac{\abs{\Tipv_{\n}(\freqx)}^2}{(\n\cdot\freqx)^2} 
  \right)^{1/2} 
				\label{int-of-hyp-x}
\end{eqnarray}
is large.
This requires that at least one of the terms in the infinite sum is
large. It is most likely that the largest term by far exceeds all
the others. So, the tail of the distribution of $\mssize$ can be understood via
the distribution of individual terms
$\termS_\n(\freqx)=\abs{\Tipv_{\n}(\freqx)}^2/(\n\cdot\freqx)^2$.
Clearly, $\P{\termS_\n>\x^2} \propto \x^{-1}$ as $\x\to+\infty$ 
as long as $\{\n\cdot\freqx=0\}\cap\Ball\ne\emptyset$,
and the distribution of $\mssize$ corresponds to the distribution 
of the square root of the largest of such terms. 
Hence, for a typical continuous distribution of $\freqx$, we expect
\begin{equation} \label{distrib-law}
  \F(\x) \equiv
  \P{\mssize>\x}\propto \x^{-1}, \quad\textrm{as}\quad \x\to+\infty. 
\end{equation}

\paragraph{Growth rate of the trajectory size.}
In practice we can observe the trajectory only for a finite, even if large,
time interval $\T$. Let us see how the expectation of the trajectory size
grows with $\T$. Consider, for instance, the departure from the initial point,
$\devn_\T$.  The exact expression for its square is 
\begin{align*}
  \abs{\devn_\T(\freqx)}^2 &= 
  \summ\limits_{\ni,\nii\in\Zahlen^{\m+1}} \frac{\cconj{\Tipv}_{\ni}\Tipv_{\nii}}{(\ni\cdot\freqx) (\nii\cdot\freqx)} 
  \\ & \mbox{} \times 
  \left( e^{-i\ni\cdot\freqx\T}-1\right) \left( e^{i\nii\cdot\freqx\T}-1\right) .
\end{align*}
Secular growth of the expectation of this series is due to resonant terms,
i.e. those with $\ni$ parallel to $\nii$. If $\vvv(\freqx)$ is smooth and 
$\Tipv_\n$ quickly decay, then the main contribution is by principal resonances $\ni=\nii$. This
gives an approximation
\[
  \abs{\devn_\T}^2 \approx  \summ\limits_{\n\in\Zahlen^{\m+1}}
  \frac{2\abs{\Tipv_{\n}}^2}{(\n\cdot\freqx)^2} 
  \left[1 - \cos\left(\n\cdot\freqx\T\right)\right] .
\]
To evaluate the corresponding expectation,
\[
  \E{\abs{\devn_\T}^2}  = \frac{1}{\mesBall}
  \int\limits_{\freqx\in\Ball}
  \abs{\devn_\T}^2
  \,\d\freqx ,
\]
where
$\mesBall=\mes{\m+1}{(\Ball)}$,
we
substitute $\z=\n\cdot\freqx\T$
and
let $\sig_\n = \mes{\m}{\left(  \{ \freqx \,|\, \n\cdot\freqx=0 \} \cap \Ball \right)}$.
This leads to
\begin{equation} \label{devn-growth}
  \E{\devn_\T^2} \approx \C_1\T, 
  \quad
  \C_1 = \frac{2\pi}{\mesBall}
  \summ\limits_{\n\in\Zahlen^{\m+1}}
  \frac{\abs{\Tipv_{\n}}^2 \sig_{\n}}{\norm{\n}}.
\end{equation}
Detailed calculations are given in the Supplementary
materials, where we also show that under similar assumptions,
\begin{equation} \label{exp-growth}
  \E{\mssizesqT\T} \approx \C_2 \T, 
  \quad
  \C_2 = \frac{\pi}{3\mesBall}
  \summ\limits_{\n\in\Zahlen^{\m+1}}
  \frac{\abs{\Tipv_{\n}}^2 \sig_{\n}}{\norm{\n}}.
\end{equation}

\dblfigure{\includegraphics{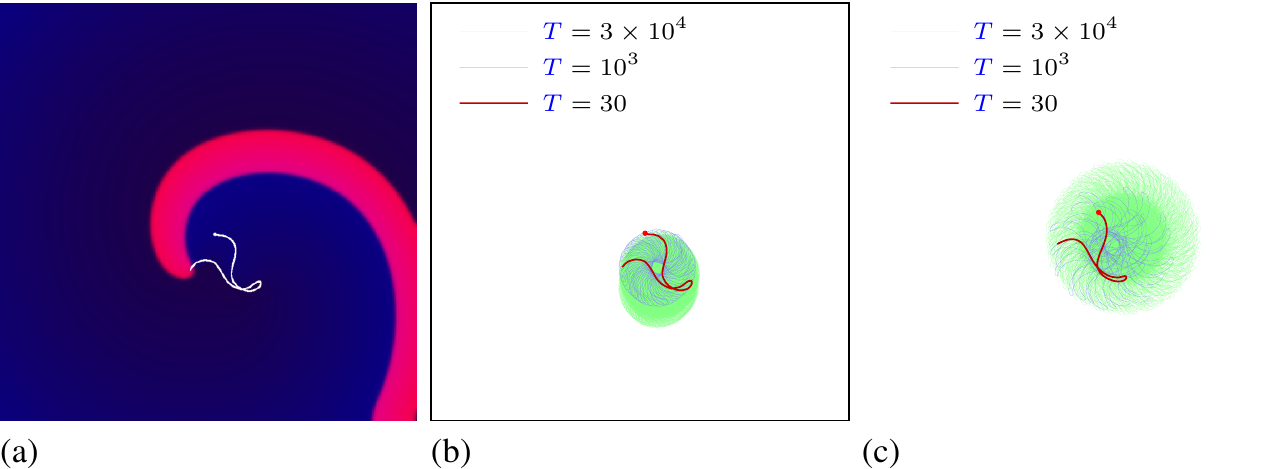}}{(colour online)
  (a) Spiral wave and a piece of meander tip trajectory in FitzHugh-Nagumo model~\eq{fhn}. 
  (b) Longer pieces of the same tip trajectory.
  (c) Pieces of trajectory of different lengths generated by the
  caricature  model~\eq{bm18} for $\freq_2=0.49777$.
}{traj}

\paragraph{Numerical illustration.} \Fig{traj}(a) shows a snapshot of
a spiral wave solution, together with a piece of the corresponding tip
trajectory, for the FitzHugh-Nagumo model \footnote{%
  This was simulated by second-order center space ($h_x=2/9$), forward
  Euler in time ($h_\t=1/125$) differencing in a box $60\times60$ with
  Neumann boundaries. The
  tip of the spiral was defined by $u=v=0$.  %
}, 
\begin{align}
 & u_\t=20(u-u^3/3-v)+\nabla^2u, \nonumber\\
 & v_\t=0.05(u+1.2-0.5v). \label{fhn}
\end{align}
\Fig{traj}(b) shows longer pieces of the tip trajectory, which
illustrates the key feature of hypermeander: the area occupied by the
trajectory can keep growing for a very long time.
We have crudely emulated these
dynamics  by a system
(\ref{tipt},\ref{qp},\ref{extension}) \footnote{%
  This was simulated with forward Euler, with $h_\t=0.3$. Crudeness of
  the method was required to make massive simulations; this does not
  affect the conclusions since this is a caricature model anyway. } with
\begin{eqnarray}
&& \m=2, \quad
   \Tipv(\phase)=\left(0.6-0.2\beta-0.2\alpha\beta\right)^{-1} - 1,
\nonumber\\
&&
   \Tiprot(\phase)
   =\left(0.675+0.1\alpha+0.05\beta+0.5\alpha^2+0.5\alpha\beta
   \right.\nonumber\\ &&\hspace{3em}\left. \mbox{}
   +0.2\alpha^3+0.6\alpha^2\beta\right)^{-1}- 1,
 \nonumber\\
&&\alpha=\cos(\phase_1)+0.05\tanh(30\cos(\phase_2)),
  \quad 
   \beta=\sin(\phase_1),
\nonumber\\
&& \freq_1=0.354, \quad \freq_2\in[0.475,0.525].	\label{bm18}
\end{eqnarray}
This was done in the spirit of~\cite{Ashwin-etal-2001} with the base
dynamics replaced by an explicit two-periodic flow, but with the
view to (i) mimic the actual meander pattern in the PDE model, and
(ii) provide sufficient nonlinearity to ensure abundance of
combination harmonics in~\eq{int-of-hyp-x}.  \Fig{traj}(c) shows
pieces of a trajectory of this ``caricature'' model. 
One can see the same key feature, that the
apparent size of the trajectory very much depends on the interval
of observation; however the details are very sensitive to the choice
of parameters, including $\freq_2$.

\dblfigure{%
   \includegraphics{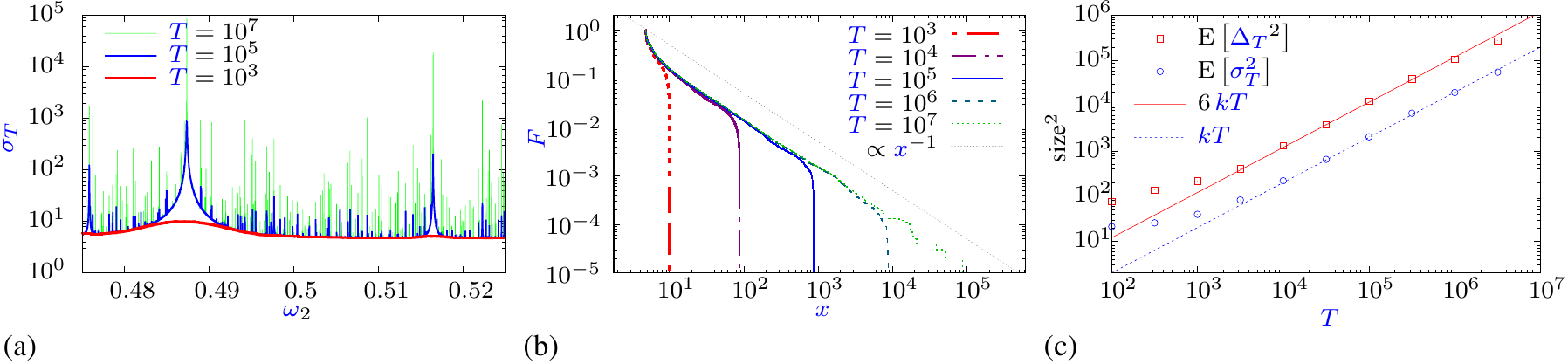}
}{(colour online)%
  (a) Sizes $\mssizeT\T$ of trajectories of different length $\T$, as
  functions of $\freq_2$, semilog plot.  %
  (b) Distribution functions $\F(\x)=\P{\mssizeT\T(\freq_2)>\x}$, for
  the trajectory sizes $\mssizeT{\T}$ for lengths $\T$, log-log plot.
  The straight line is the asymptotic~\eq{distrib-law}. %
  (c) Mean square of trajectory size as function of the time interval, log-log plot.
  The straight lines are the asymptotics~\eq{devn-growth} and
    \eq{exp-growth} with a fitted $\K=\C_2=\C_1/6$.
}{stats}

\Fig{stats}(a) illustrates the approach of $\mssizeT{\T}(\freq_2)$ to an
everywhere discontinuous function as $\T\to\infty$. This was obtained
for $10^5$ values of $\freq_2$ randomly chosen in the shown interval.
For smaller $\T$ one can see well shaped individual
peaks associated with the poles of $\mssize(\freq_2)$ corresponding to
the resonances with the highest $\Tipv_{\n}$; for larger $\T$, more of
such peaks become pronounced, and they grow stronger. 
\Fig{stats}(b) shows the empirical distribution of the
trajectory sizes for the pieces of trajectories of the same $10^5$
simulations, of different lengths. We see that for larger $\T$, the
distribution approaches the theoretical prediction~\eq{distrib-law}.
Finally, \fig{stats}(c) shows the growth of two empirical estimates of
trajectory size with time, in agreement with~\eq{exp-growth}
and~\eq{devn-growth}, including the predicted approximate ratio of 6 between them.

\paragraph{In conclusion,} 
quasiperiodic hypermeander of spiral waves has paradoxical properties. Even
though described by deterministic equations, with no chaos involved, the question of the size
of the tip trajectory does not have a meaningful deterministic answer
and requires probabilistic treatment. In probabilistic terms, although
the tip trajectory is confined with probability one, the expectation
of its size, however measured, is infinite. 
There it is similar to the ``St. Petersburg lottery'', in which a win
is almost certainly finite, but its expectation is
infinite~\cite{Bernoulli-1738,Sommer-1954}.  The realistic price for
a ticket in this lottery is nevertheless finite and modest; the
resolution of this pardox relevant to us is that high wins require
unrealistically long games~\cite[Section X.4]{Feller-1968}.
In our
case, the dependence of the trajectory size,
whether defined via mean square displacement $\abs{\devn_\T}^2$ or
  variance $\mssizesqT\T$,
on any parameter affecting
the frequency ratios becomes more and more irregular as $\T\to\infty$,
and the expectations $\E{\abs{\devn_\T}^2}$ and
  $\E{\mssizesqT\T}$ defined as averages over parameter variations,
  grow linearly in $\T$ even though the individual trajectories are
  bounded.  Note that this is different from the linear growth for the
  mean square displacement of chaotically hypermeandering spirals
  \cite{Biktashev-Holden-1998,Ashwin-etal-2001} which is for averages over initial conditions.

Practical applications of the theory are most evident for re-entrant
waves in cardiac tissue, underlying dangerous cardiac arrhythmias.
However implications may be also expected in any physics where the
theory involves differential equations with quasiperiodic
coefficients.  One example may be provided by evolution of tracers in
quasi-periodic fluid flows~\cite{Boatto-Pierrehumbert-1999}. On a more
speculative level, extension from ODEs in time to PDEs in spatial
variables may provide insights into properties of
quasicrystals~\cite{Janssen-Janner-2015} or quasiperiodic dissipative
structures~\cite{Subramanian-etal-2016}.  
Note that properties of
quasicrystals, among other things, include
superlubricity~\cite{Koren-Duerig-2016} and
superconductivity~\cite{Kamiya-etal-2018}, still awaiting full
theoretical treatment.

\emph{Acknowledgements} %
  VNB was supported by EPSRC Grant
  EP/N014391/1 
  (UK), 
  NSF Grant PHY-1748958, %
  NIH Grant R25GM067110, %
  and the Gordon and Betty Moore Foundation Grant 2919.01 %
  (USA).
IM was supported by European Advanced Grant ERC AdG 320977 (EU).

\onecolumngrid
\cleardoublepage
\appendix
\onecolumngrid
\makeatletter
\begin{centering}
\section{
  Supplementary material for \\ 
  ``\@title'' \\ 
  by V. N. Biktashev and I. Melbourne
  \@author
}
\end{centering}
\makeatother

\subsection{Details of the derivation of the trajectory size asymptotics}

The exact expression for the square of the departure from the initial point is
\[
\abs{\devn_\T(\freqx)}^2 = \abs{ \tipp(\T) - \tipp(0) }^2
 = \abs{ 
   \summ\limits_{\n\in\Zahlen^{\m+1}} \frac{-i\Tipv_{\n}}{\n\cdot\freqx} 
   \left( e^{i\n\cdot\freqx\T}-1\right) 
 }^2 .
\]
Then for its expectation we have
\[
  \E{\abs{\devn_\T}^2} = \frac{1}{\mesBall} \int\limits_\Ball 
  \abs{\devn_\T(\freqx)}^2
  \,\d\freqx
  = \frac{1}{\mesBall} 
  \summ\limits_{\ni,\nii} 
  \Tipv_{\ni}\cconj{\Tipv}_{\nii}
  \termD_{\ni\nii}(\T),
\]
where
\[
  \termD_{\ni\nii}(\T)
  =
  \int\limits_\Ball 
  \frac{e^{-i\ni\cdot\freqx\T}-1}{\ni\cdot\freqx}
  \, 
  \frac{e^{i\nii\cdot\freqx\T}-1}{\nii\cdot\freqx} 
  \,\d\freqx .
\]
Let us investigate the behaviour of the coefficients  $\termD_{\ni\nii}(\T)$ in the limit $\T\to\infty$. 
We have to consider separately the cases when the two zero-denominator hyperplanes cut or do not cut through $\Ball$. 
Recall that
\(
  \sig_\n \equiv \mes{\m}{ \{ \freqx \,|\, \freqx\in\Ball \;\&\;
  \n\cdot\freqx=0 \}}
\), 
and
\(
  \norm{\n} \equiv \left( \sum\limits_{\j=1}^{\m+1} \n_\j^2 \right)^{1/2}
\).
We write $\n'\parallel\n''$ when vectors $n'$ and $n''$ are
  parallel (linearly dependent), and $\n'\nparallel\n''$ otherwise.

\begin{itemize}
\item For $\sig_{\ni}=0$, $\sig_{\nii}=0$, the coefficients are bounded:
\[
  \abs{\termD_{\ni\nii}(\T)} \le
   4\mesBall
   \left(\min\limits_{\freqx\in\Ball}\abs{\ni\cdot\freqx}\right)^{-1}
   \left(\min\limits_{\freqx\in\Ball}\abs{\nii\cdot\freqx}\right)^{-1}
   =\O{1} .
\]
\item For $\sig_{\nii}=0$, $\sig_{\ni}\ne0$, we use a change of
  variables in the space $\left\{\freqx\right\}=\Real^{\m+1}$; namely,
  $\z=\ni\cdot\freqx\T\in\Real$, and $\ziii\in\Real^{\m}$ for the
  unscaled coordinates in $\ni^\perp$. 
  In coordinates $(\z,\ziii)$, 
  the domain $\Ball$ is stretched
  in the $\z$ direction and, as $\T\to\infty$, tends to an infinite cylinder
  with the axis along the $\z$ axis and the base of measure $\sig_\ni$.
  This gives
  \[
    \abs{\termD_{\ni\nii}(\T)} \le
    \left(\min\limits_{\freqx\in\Ball}\abs{\nii\cdot\freqx}\right)^{-1}    
    \int\limits_\Ball \abs{
    \frac{e^{i\ni\cdot\freqx\T}-1}{\ni\cdot\freqx}
    }\,\d\freqx
    = 
    \left(\min\limits_{\freqx\in\Ball}\abs{\nii\cdot\freqx}\right)^{-1}
    \iint\limits_{\freqx\in\Ball} \abs{
      \frac{e^{i\z}-1}{\z/\T}
    }
    \frac{\d\z}{\norm{\ni}\T}
    \,\d\ziii
  \]\[
    =
    \left(\min\limits_{\freqx\in\Ball}\abs{\nii\cdot\freqx}\right)^{-1}
    \frac{\sig_{\ni}}{\norm{\ni}}
      \int\limits_{-\O{\T}}^{\O{\T}} \abs{\frac{e^{i\z}-1}{\z}} \,\d\z
  = \O{\ln(\T)}, 
\]
and similarly for $\sig_{\ni}=0$, $\sig_{\nii}\ne0$.  
\item For $\sig_{\ni}\ne0$, $\sig_{\nii}\ne0$, and $\ni
    \nparallel \nii$, 
we use variables
$\zi=\ni\cdot\freqx\T\in\Real$, 
$\zii=\nii\cdot\freqx\T\in\Real$, 
and $\ziii\in\Real^{\m-1}$ for the unscaled coordinates in
$\Span(\ni,\nii)^\perp$. 
Then 
\[
  \termD_{\ni\nii} (\T)= \iiint\limits_{\freqx\in\Ball} 
  \dfrac{e^{-i\zi}-1}{\zi/\T} 
  \,
  \dfrac{e^{i\zii}-1}{\zii/\T} 
  \,
  \dfrac{\d\zi \d\zii}{\norm{\ni}\norm{\nii}\sin\left(\widehat{\ni,\nii}\right) \T^2}
  \,
  \d\ziii
  =\O{1}. 
\]
Here $\widehat{\ni,\nii}$ is the angle between vectors $\ni$ and
  $\nii$.
\item For $\sig_{\ni}\ne0$, $\sig_{\nii}\ne0$, and $\ni
  \parallel \nii$, we set $\ni=\ai\n$, $\nii=\aii\n$, where
  $\n\in\Zahlen^{\m+1}$ is their GCD vector and
  $\ai,\aii\in\Zahlen\setminus\{0\}$ (see Proposition~\ref{GCD} below).
  In this case the hyperplanes $\ni\cdot\freqx=0$,
  $\nii\cdot\freqx=0$ and $\n\cdot\freqx=0$ coincide, and
  correspondingly $\sig_\ni=\sig_\nii=\sig_\n$. 

We use
$\z=\n\cdot\freqx\T\in\Real$, and $\ziii\in\Real^{\m}$
for the unscaled coordinates in $\n^\perp$.
That gives

\[
  \termD_{\ni\nii} (\T)
  =
  \iint\limits_{\freqx\in\Ball} 
  \frac{\left(e^{i\ai\z}-1\right) \left(e^{-i\aii\z}-1\right)}{\ai\aii (\z/\T)^2} 
  \,\dfrac{\d\z \d\ziii}{\norm{\n}\T}
  =
  \frac{\T\sig_\n}{\ai\aii\norm{\n}}
  \int\limits_{-\O{\T}}^{\O{\T}}
  \left[
    e^{i(\ai-\aii)\z}
    -e^{i\ai\z}
    -e^{-i\aii\z}
    +1 
  \right]
  \frac{\d\z}{\z^2} .
\]
Now,
\begin{equation}\label{Dint}
  \intinf
  \left[
    e^{i(\ai-\aii)\z}
    -e^{i\ai\z}
    -e^{-i\aii\z}
    +1 
  \right]
  \frac{\d\z}{\z^2}
  =
  \I(\ai) + \I(\aii) - \I(\ai-\aii) ,
\end{equation}
where
\begin{equation}\label{Iint}
  \I(\alpha)=\intinf (1-\cos(\alpha\z)) \,\frac{\d{\z}}{\z^2}
  = \pi \abs{\alpha} ,
\end{equation}
and therefore
\[
  \termD_{\ni\nii} (\T)
  \approx
  \frac{ \pi \sig_\n \left(\abs{\ai} + \abs{\aii} - \abs{\ai-\aii}\right) }{\norm{\n} \ai\aii} \T
  \quad
  \textrm{as}
  \quad
  \T\to\infty.
\]
For the principal resonances
$\ai=\aii=1$ we have
\[
  \termD_{\n\n} (\T)
  \approx
  2\pi
  \frac{\sig_\n}{\norm{\n}} \T
  \quad
  \textrm{as}
  \quad
  \T\to\infty ,
\]
giving the estimate \eq{devn-growth}.
\end{itemize}

The computations for other statistics are similar in technique,
if slightly longer. 
The raw second moment, i.e. the expectation of the 
time-average of the square departure from initial point, is 
\[
  \E{\mspathsqT\T} = 
  \frac{1}{\T\mesBall}\int\limits_{\Ball} \int\limits_{0}^{\T}
  \abs{\devn_\t(\freqx)}^2 \,\d\t\,\d\freqx
  = 
  \frac{1}{\mesBall} 
  \summ\limits_{\ni,\nii\in\Zahlen^{\m+1}}
  \Tipv_{\ni}\cconj\Tipv_{\nii}
  \termP_{\ni,\nii}(\T),
\]
where, for $\ni\nparallel\nii$, 
\[
  \termP_{\ni\nii}(\T)
  =
  \int\limits_\Ball 
  \frac{1}{(\ni\cdot\freqx)(\nii\cdot\freqx)}
  \left[
    \frac{e^{i(\ni-\nii)\cdot\freqx \T}-1}{i(\ni-\nii)\cdot\freqx \T}
    -\frac{e^{i\ni\cdot\freqx\T}-1}{i\ni\cdot\freqx \T}
    -\frac{e^{-i\nii\cdot\freqx\T}-1}{-i\nii\cdot\freqx \T}
    +1
  \right]
  \,
  \,\d\freqx ,
\]

and for $\ni\parallel\nii$, $\ni/\ai=\nii/\aii=\n$, 

\[
  \termP_{\ni\nii}(\T)
    = \frac{1}{\ai\aii\T}
    \left[ \I(\ai)+\I(-\aii)-\I(\ai-\aii) \right]
\]
where
\[
   \I(\al) = \int\limits_{\Ball}
    \frac{
      1 + i\al\n\cdot\freqx\T - \left( \al\n\cdot\freqx\T \right)^2/2
      - e^{i\al\n\cdot\freqx\T}
    }{
      i \al(\n\cdot\freqx)^3
    }
  \,\d\freqx . 
\]

Reasoning as in the previous case, we conclude that 
all the terms are
$\O{\ln(\T)}$
as $\T\to\infty$, except for 
those with $\ni\parallel\nii$,
$\sig_\n\ne0$, which grow as $\O{\T}$. 
Using, as before, the variables
$\z=\n\cdot\freqx\T\in\Real$ and $\ziii\in\Real^{\m}\sim\n^\perp$,
we get

\[
  \I(\al) 
  =
  \frac{\sig_\n\T^2\abs{\al}}{\norm{\n}} \int\limits_{-\O{\T}}^{\O{\T}}
    \frac{
      1 + i\al\z - \left( \al\z \right)^2/2 - e^{i\al\z}
    }{
      i\z^3
    }
    \, \d\z
  \approx
  \frac{\sig_\n\T^2\abs{\al}}{\norm{\n}} \intinf
  \frac{\z - \sin(\z)}{\z^3}
  \, \d\z
  = \frac{\sig_\n\T^2\abs{\al}}{\norm{\n}} \frac{\pi}{2},
\]
so
\[
  \termP_{\ni\nii} \approx
  \frac{\pi\sig_\n( \abs{\ai}+\abs{\aii}-\abs{\ai-\aii})}{2\norm{\n}\ai\aii} 
  \T
  \quad
  \textrm{as}
  \quad
  \T\to\infty.
\]
For the principal resonances, $\ai=\aii=1$, this simplifies to
\[
  \termP_{\n\n} \approx
  \pi \frac{\sig_\n}{\norm{\n}} 
  \T
  \quad
  \textrm{as}
  \quad
  \T\to\infty.
\]

The expectation of the square of the time-averaged departure from
initial point, i.e. of the length of the position vector of the apparent centroid in
time $\T$,  is
\[
  \E{ \abs{\meansizeT{\T}}^2 }
  =
  \frac{1}{\mesBall} \int\limits_\Ball \abs{
    \frac1\T \int\limits_0^\T \devn_\T(\freqx) \,\d\T
  }^2 \,\d\freqx
  = 
  \frac{1}{\mesBall} 
  \summ\limits_{\ni,\nii\in\Zahlen^{\m+1}}
  \Tipv_{\ni}\cconj\Tipv_{\nii}
  \termM_{\ni,\nii}(\T),
\]
%
where
\[
  \termM_{\ni\nii}(\T)=
    \int\limits_{\Ball}
    \frac{
      \left(e^{i\ni\cdot\freqx\T}-1-i \ni\cdot\freqx\T\right)
      \left(e^{-i\nii\cdot\freqx\T}-1+i \nii\cdot\freqx\T\right)
    }{
      (\ni\cdot\freqx)^2(\nii\cdot\freqx)^2\T^2
    }
    \,\d\freqx .
\]
As before, important terms are those with $\ni/\ai=\nii/\aii=\n$, 
$\sig_\n\ne0$,
for which we
use $\z=\n\cdot\freqx\T$, $\ziii\in\Real^{\m}\sim\n^\perp$, and
get

\[
  \termM_{\ni\nii}(\T) \approx 
  \frac{\sig_\n\T}{\norm{\n}\ai^2\aii^2}
  \termM(\ai,\aii)
  \quad
  \textrm{as}
  \quad
  \T\to\infty,
\]
where the integral
\[
  \termM(\ai,\aii)
  = 
  \intinf
      \left(e^{i\ai\z}-1-i \ai\z\right)
      \left(e^{-i\aii\z}-1+i \aii\z\right)
    \,\frac{\d\z}{\z^4}
\]
can be calculated using differentiation by parameters. We have
\[
  \dff{\termM}{\ai}{\aii}
  =
  \intinf
      \left(e^{i\ai\z}-1\right)
      \left(e^{-i\aii\z}-1\right)
    \,\frac{\d\z}{\z^2}
    =
    \pi\left( \abs{\ai}+\abs{\aii}-\abs{\ai-\aii} \right),
\]
using the result \eqtwo(Dint,Iint) obtained above. 
Hence,
\[
  \termM(\ai,\aii)
  =
  \pi \iint \left( \abs{\ai}+\abs{\aii}-\abs{\ai-\aii} \right) \,\d\ai\d\aii
\]\[
  = \pi\left(
    \frac12\aii\ai\abs{\ai}
    +
    \frac12\ai\aii\abs{\aii}
    +
    \frac16(\ai-\aii)^2\abs{\ai-\aii}
    \right)
    + \arbF(\ai) + \arbG(\aii),
\]
where functions $\arbF$ and $\arbG$ can be determined from boundary conditions.
Consider
\[
  \termM(\ai,0)=0
  =
    \frac{\pi}{6}\ai^2\abs{\ai}
    + \arbF(\ai) + \arbG(0),
\]\[
  \termM(0,\aii)=0
  = 
  \dfrac{\pi}{6}\aii^2\abs{\aii}
  + \arbF(0) + \arbG(\aii) .
\]
We observe that $\arbF(0)=\arbG(0)=0$ is an admissible choice, which leads to
\[
  \termM(\ai,\aii)
  = \frac{\pi}{6}\left(
    (3\aii-\ai)\ai\abs{\ai}
    +
    (3\ai-\aii)\aii\abs{\aii}
    +
    (\ai-\aii)^2\abs{\ai-\aii}
    \right)
\]
and consequently
\[
  \termM_{\ni\nii}(\T)\approx
  \frac16 \pi\sig_\n
  \frac{(3\aii-\ai)\ai\abs{\ai}
    +
    (3\ai-\aii)\aii\abs{\aii}
    +
    (\ai-\aii)^2\abs{\ai-\aii}
  }{\norm{\n}\ai^2\aii^2}
  \T
  \quad\textrm{as}\quad\T\to\infty.
\]
For the principal resonances, $\ai=\aii=1$, this gives
\[
  \termM_{\ni\nii}(\T) \approx
  \frac{2\pi}{3}
  \,
  \frac{\sig_\n}{\norm\n}
  \T
  \quad\textrm{as}\quad\T\to\infty.
\]
Hence the central second moment, i.e. the expectation of the 
time-average of the square departure from the apparent centroid is
\[
  \E{\mssizesqT\T}
  =
  \frac{1}{\mesBall} \int\limits_\Ball \abs{
    \frac1\T \int\limits_0^\T \left(
      \devn_\T(\freqx) 
      -
      \meansizeT{\T}(\freqx)
    \right)\,\d\T
  }^2 \,\d\freqx
  =
  \E{\mspathsqT\T   - \abs{\meansizeT{\T}}^2}
  =
  \frac{1}{\mesBall} 
  \summ\limits_{\ni,\nii\in\Zahlen^{\m+1}}
  \Tipv_{\ni}\cconj\Tipv_{\nii}
  \termS_{\ni,\nii}(\T),
\]
where for the principal resonances we have 
\[
  \termS_{\n\n}(\T)=
  \termP_{\n\n}(\T)-
  \termM_{\n\n}(\T)
  \approx
  \frac{\pi}{3}
  \,
  \frac{\sig_\n}{\norm{\n}}
  \T ,
\]
which gives the estimate~\eq{exp-growth}.

\begin{proposition}\label{GCD}
  Let $\ni,\nii\in\Zahlen^\m\setminus\{0\}$ be linearly
  dependent. Then there exist $\ai,\aii\in\Zahlen\setminus\{0\}$ and
  $\n\in\Zahlen^\m\setminus\{0\}$ such that $\ai$ and $\aii$ are coprime and $\ni=\ai\n$, $\nii=\aii\n$.
\end{proposition}
\begin{proof}
  Since both
  vectors are nonzero, we have $\nii=\al\ni$ for a nonzero scalar
  $\al$. We must have $\al\in\Quoziente\setminus\{0\}$ since it is a ratio of
  the corresponding components of $\ni$ and $\nii$. Let $\al=\aii/\ai$
  with $\ai,\aii\in\Zahlen\setminus\{0\}$ coprime. By writing $\ai\nii=\aii\ni$ we observe
  that all components of $\nii$ are divisible by $\aii$ and all
  components of $\ni$ are divisible by $\ai$. Hence
  $\nii/\aii=\ni/\ai=\n\in\Zahlen^\m\setminus\{0\}$, as required. 
\end{proof}
\end{document}